\documentclass[journal,twoside]{IEEEtran}
\usepackage{graphicx,epsfig,multirow,array}
\usepackage{bm}
\usepackage{subfigure}
\usepackage{amsmath}
\newtheorem{theorem}{Theorem}

\newtheorem{proof}{Proof}
\usepackage{amssymb}
\usepackage{cite}
\usepackage{stfloats}
\usepackage{multicol}
\usepackage{algorithm}
\usepackage{algorithmicx}
\usepackage{algpseudocode}
\usepackage{hhline}
\usepackage{colortbl}
\usepackage{color}

\usepackage{mathrsfs}
\usepackage{bm}

\begin{document}

\title{Secret Key Generation for Intelligent Reflecting Surface Assisted Wireless Communication Networks}


\author {Zijie~Ji,
        Phee~Lep~Yeoh,~\IEEEmembership{Member,~IEEE,}
        Deyou~Zhang,
        Gaojie~Chen,~\IEEEmembership{Senior Member,~IEEE,}
        Yan~Zhang,~\IEEEmembership{Member,~IEEE,}
        Zunwen~He,~\IEEEmembership{Member,~IEEE,}
        Hao~Yin,
        and~Yonghui Li,~\IEEEmembership{Fellow,~IEEE}
\thanks{This work was supported by the National Key R\&D Program of China under Grant 2020YFB1804901, the National Natural Science Foundation of China under Grant 61871035, the China Scholarship Council scholarship, and Ericsson company. \emph{(Corresponding author: Yan Zhang.)}
Copyright (c) 2015 IEEE. Personal use of this material is permitted. However, permission to use this material for any other purposes must be obtained from the IEEE by sending a request to pubs-permissions@ieee.org.}
\thanks{Z.~Ji, Y.~Zhang, and Z.~He are with the School of Information and Electronics, Beijing Institute of Technology, Beijing 100081, China (e-mail: \{jizijie, zhangy, hezunwen\}@bit.edu.cn).}
\thanks{P.~L.~Yeoh, D.~Zhang, and Y.~Li are with the School of Electrical and Information Engineering, University of Sydney, Sydney, NSW 2006, Australia (e-mail: \{phee.yeoh, deyou.zhang, yonghui.li\}@sydney.edu.au).}
\thanks{G.~Chen is with the Department of Engineering, University of Leicester, Leicester LE1 7RH, U.K. (e-mail: gaojie.chen@leicester.ac.uk).}
\thanks{H.~Yin is with Institute of China Electronic System Engineering Corporation, Beijing 100141, China (e-mail: yinhao@cashq.ac.cn).}
}

\markboth{IEEE TRANSACTIONS ON VEHICULAR TECHNOLOGY}
{Ji \MakeLowercase{\textit{et al.}}: Secret Key Generation for Intelligent Reflecting Surface Assisted Wireless Communication Networks}
\maketitle

\begin{abstract}
We propose and analyze secret key generation using intelligent reflecting surface (IRS) assisted wireless communication networks. To this end, we first formulate the minimum achievable secret key capacity for an IRS acting as a passive beamformer in the presence of multiple eavesdroppers. Next, we develop an optimization framework for the IRS reflecting coefficients based on the secret key capacity lower bound. To derive a tractable and efficient solution, we design and analyze a semidefinite relaxation (SDR) and successive convex approximation (SCA) based algorithm for the proposed optimization. Simulation results show that employing our IRS-based algorithm can significantly improve the secret key generation capacity for a wide-range of wireless channel parameters.
\end{abstract}

\begin{IEEEkeywords}
Intelligent reflecting surface, physical layer security, secret key generation, semidefinite relaxation, successive convex approximation.
\end{IEEEkeywords}
\IEEEpeerreviewmaketitle
\section{Introduction}
\IEEEPARstart{I}{ntelligent} reflecting surface (IRS) is a promising emerging communication architecture for future wireless networks, which enables smart reconfiguring of the signal propagation environment by using passive reflecting elements with controllable amplitudes and/or phase shifts \cite{reference1}, \cite{reference1.1}. By applying large-scale passive beamforming signal processing, IRS has been shown to effectively improve the data transmission performance by combating deleterious wireless channel conditions such as co-channel interference and dead-zones \cite{reference2}. Recently, IRS has also been considered for improving the physical layer security in wireless communication networks \cite{reference3, reference4,reference4.1}. Most of these works focused on investigating the information-theoretic secrecy transmission, where beamforming and artificial noise vectors are designed to maximize the difference in signal-to-noise ratio (SNR) between legitimate channels and eavesdropping channels, and have not considered the use of IRS for wireless secret key generation.

Secret key generation is a lightweight physical layer security approach which allows legitimate users to establish shared keys based on the correlation between their reciprocal channels. In such scenarios, it is challenging for the eavesdropper to acquire information about the generated keys since there is typically low correlation between the eavesdropper channels and the legitimate channels. In \cite{reference5}, the high directionality and sparsity of millimetre-wave wireless channels were exploited to prevent attacks from co-located eavesdroppers by employing large-scale active beamforming at the transmitter. In \cite{reference6}, the authors considered secret key generation with untrusted relays, where the secret key capacity was optimized independently for the source-to-relay and relay-to-destination links. Recently in \cite{reference7}, we showed that a high scattering multipath channel environment can significantly improve the security of generated keys.

In this paper, inspired by the aforementioned works, we consider the use of IRS as a new degree of freedom (DoF) for wireless secret key generation where each element is an individual scatterer to boost the secret key capacity. Unlike secrecy transmission, the aim of secret key generation is to increase the correlation between legitimate uplink and downlink channels while reducing their correlation with eavesdropping channels. Due to the use of passive beamforming at the IRS, a major challenge for secret key generation is the need to jointly optimize both of the IRS links to the legitimate users. Different from previous works, we assume that the legitimate users are low-cost single antenna devices and the main security advantage is from the large-scale IRS. Our main contributions are as follows:
\begin{itemize}
  \item We derive a new closed-form expression for a lower bound on the secret key capacity of IRS assisted wireless networks with multiple non-colluding eavesdroppers. Our analytical expression accurately characterizes the impact of the channel correlations between the legitimate users, eavesdroppers and IRS.
  \item We develop an optimization framework for the IRS coefficients based on our analytical lower bound expression that maximizes the minimum secret key capacity for the worst-case eavesdropper channel.
  \item To overcome the non-convexity of the objective function, we derive a low-complexity algorithm by applying semidefinite relaxation (SDR) and successive convex approximation (SCA) for the IRS coefficient matrix.
\end{itemize}
Simulations validate that the proposed algorithm can significantly improve the secret key capacity, and we highlight the importance of carefully optimizing the IRS coefficients based on the line-of-sight (LoS) channel conditions and eavesdropper locations.

\section{System Model}
Fig. \ref{figure1} shows an IRS assisted wireless communication network, where a single-antenna wireless access point (Alice) and a single-antenna mobile user (Bob) want to generate shared secret keys based on their reciprocal wireless channels, while $K$ single antenna non-colluding eavesdroppers (Eves) attempt to access the secret keys generated by Alice and Bob based on their own channel observations. The secret key generation between Alice and Bob is assisted by Rose, who is an $N$-element IRS that can modify her reflecting coefficients to minimize the secret key leakage to Eves.

\begin{figure}[t]
  \centering
  \includegraphics[width=0.4\textwidth]{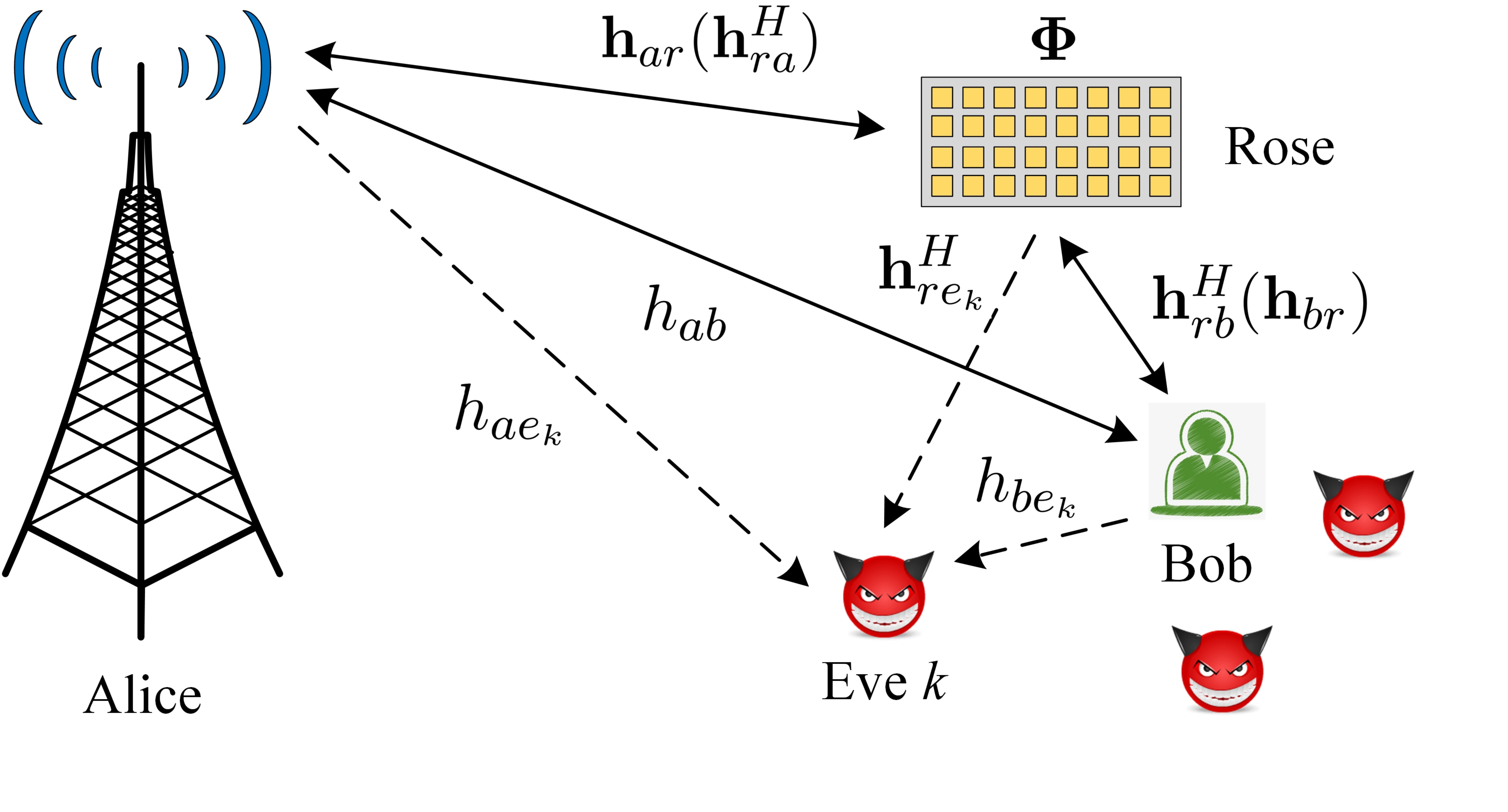}\\
  \caption{System model for an IRS assisted wireless communication network.}
  \label{figure1}
\end{figure}

To generate their shared secret keys, we assume that Alice and Bob will alternatively exchange pilots and perform channel estimations in time-division duplex (TDD) mode. In the odd time slots, Alice transmits pilot signal ${\bf{s}}_1$, and the received signal at Bob or the $k$th Eve is given by
\begin{equation}\label{equation1}
{{\bf{y}}_{i,1}} = ({{\tilde h}_{ai}} + {\bf{{\tilde h}}}_{ri}^H{\bf{\Phi }}{{\bf{{\tilde h}}}_{ar}}){{\bf{s}}_1} + {{\bf{z}}_i},\ i \in \{ b,{e_k}\},
\end{equation}
where ${{\tilde h}_{ai}}\in {\mathbb{C}^{1 \times 1}}$, ${{\bf{{\tilde h}}}_{ar}}\in {\mathbb{C}^{N \times 1}}$, and ${\bf{{\tilde h}}}_{ri}\in {\mathbb{C}^{N \times 1}}$ are the direct channel from Alice to node $i$, the incident channel from Alice to Rose, and the reflected channel from Rose to Bob or Eve $k$, respectively.
For the IRS channel, similar to \cite{reference1} and \cite{reference2}, we assume that ${\bf{\Phi }} = {\rm{diag}}({\alpha _1}{e^{j{\theta _1}}},{\alpha _2}{e^{j{\theta _2}}}, \ldots ,{\alpha _N}{e^{j{\theta _N}}})$ denotes the diagonal amplitude-phase shifting reflecting coefficient matrix of Rose, where ${\alpha _n} \!\in\! [0,1]$ and ${{\theta _n}} \!\in\! [0,2\pi )$ are the amplitude and phase shifts on the incident signal by its $n$th element, $n=1, \ldots ,N$. We denote ${\bf{z}_i} \sim \mathcal{CN}(0,\sigma _i^2\bf{I})$ as the independent and identically distributed (i.i.d.) complex additive white Gaussian noise vector, where $\bf I$ is the identity matrix. Similarly, once ${{\bf{y}}_{b,1}}$ is received, Bob sends pilot signal ${\bf{s}}_2$ in the even time slots, thus Alice or Eve $k$ receives
\begin{equation}\label{equation2}
{{\bf{y}}_{i,2}} = ({{\tilde h}_{bi}} + {\bf{{\tilde h}}}_{ri}^H{\bf{\Phi }}{{\bf{{\tilde h}}}_{br}}){{\bf{s}}_2} + {{\bf{z}}_i},\ i \in \{ a,{e_k}\},
\end{equation}
where ${{\tilde h}_{bi}}\in {\mathbb{C}^{1 \times 1}}$ and ${{\bf{{\tilde h}}}_{br}}\in {\mathbb{C}^{N \times 1}}$ denote the channel from Bob to node $i$ and Rose, respectively. We assume that the pilot symbols ${\bf{s}}_1$ and ${\bf{s}}_2$ have zero mean and unit variance, and are known by all nodes. We also assume that the sampling interval is sufficiently small such that the channel reciprocity holds between the bidirectional transmissions \cite{reference8}. Therefore, without loss of generality, the combined channels can be estimated via channel estimation algorithms, e.g., least square (LS), at each node (Alice, Bob, and Eve $k$) as
\begin{equation}\label{equation3}
\begin{array}{*{20}{l}}
{{{\tilde h}_B} = {\bf{y}}_{i,1}^H{{\bf{s}}_1}/||{{\bf{s}}_1}||_2^2 = ({{\tilde h}_{ab}} + {\bf{\tilde h}}_{rb}^H{\bf{\Phi }}{{{\bf{\tilde h}}}_{ar}}) + {{\hat z}_b},}\\
{{{\tilde h}_A} = {\bf{y}}_{i,2}^H{{\bf{s}}_2}/||{{\bf{s}}_2}||_2^2 = ({{\tilde h}_{ba}} + {\bf{\tilde h}}_{ra}^H{\bf{\Phi }}{{{\bf{\tilde h}}}_{br}}) + {{\hat z}_a},}\\
{{{\tilde h}_{E1_k}} = {\bf{y}}_{i,1}^H{{\bf{s}}_1}/||{{\bf{s}}_1}||_2^2 = ({{\tilde h}_{a{e_k}}} + {\bf{\tilde h}}_{r{e_k}}^H{\bf{\Phi }}{{{\bf{\tilde h}}}_{ar}}) + {{\hat z}_{{e_k}}},}\\
{{{\tilde h}_{E2_k}} = {\bf{y}}_{i,2}^H{{\bf{s}}_2}/||{{\bf{s}}_2}||_2^2 = ({{\tilde h}_{b{e_k}}} + {\bf{\tilde h}}_{r{e_k}}^H{\bf{\Phi }}{{{\bf{\tilde h}}}_{br}}) + {{\hat z}_{{e_k}}},}
\end{array}
\end{equation}
where $\lVert\cdot\rVert_2$ is the Euclidean norm, and ${\hat z_i} \sim \mathcal{CN}(0,\sigma _i^2),\ i \in \{ a,b,{e_k}\}$ denotes the estimation error. $\tilde h_{E1_k}$ and $\tilde h_{E2_k}$ are the estimated channels in the odd and even time slots, respectively. After $L$ rounds of pilot exchanges within a single coherence time $T_c$, the combined channel vectors are denoted as ${{\bf{\tilde h}}_D}\!=\![{\tilde h_D}(1), {\tilde h_D}(2), \ldots ,{\tilde h_D}(L)]$, where $D \!\in\! \{ A,B,E1_k,E2_k\}$. The $L$ samples are collected to obtain reliable channel statistics such as covariance and correlation for subsequent optimization processing. Furthermore, all channel vectors are normalized as ${{\bf{h}}_D} = {{\bf{\tilde h}}_D}/{\lVert {{{\bf{\tilde h}}_D}} \rVert_2}$ before quantization to eliminate the impact of amplitude difference.
\newcounter{TempEqCnt}
\setcounter{TempEqCnt}{\value{equation}} 
\setcounter{equation}{10} 
\begin{figure*}[hb] 
	\hrulefill  
	\begin{equation}\label{equation11}
	{C_{{\rm{lb1}},k}} = \frac{1}{{{T_c}}}{\log _2}\left(1 + \frac{{K_{ab}^2\sigma _{{e_k}}^4 - {K_{ab}}({K_{a{e_k}}} + {K_{b{e_k}}})\sigma _b^2\sigma _{{e_k}}^2}}{{(({K_{a{e_k}}} + {K_{b{e_k}}})\sigma _{{e_k}}^2 + \sigma _{{e_k}}^4)({K_{ab}}(\sigma _a^2 + \sigma _b^2) + \sigma _a^2\sigma _b^2)}}\right).
	\end{equation}
\end{figure*}
\setcounter{equation}{\value{TempEqCnt}}

\section{IRS Secret Key Capacity Lower Bound}
In this section, we derive a new closed-form expression for a lower bound on the secret key capacity of IRS assisted wireless networks. Here, as obtaining the exact secret key capacity is still an open problem, we consider the lower bound capacity as the minimal achievable secret key rate which is more conservative compared with the upper bound. For general wireless channels, the lower bound on the secret key capacity can be expressed as \cite{reference6}
\begin{equation}\label{equation4}
\begin{aligned}
&C({h_A};{h_B}\!\parallel\! {h_{E1_k}},{h_{E2_k}}) \ge
\max \{ I({h_A};{h_B})\\
&- I({h_A};{h_{E1_k}},{h_{E2_k}}),
I({h_A};{h_B}) \!-\! I({h_B};{h_{E1_k}},{h_{E2_k}})\},
\end{aligned}
\end{equation}
where $I(X;Y)$ is the mutual information of variables $X$ and $Y$, and max \{$\cdot$\} is the maximum function. To derive the minimum
achievable secret key capacity in IRS assisted wireless communication networks, we assume that the statistics of the channel state information (CSI) of all channels are known at Alice and Rose. This is a common assumption in large-scale wireless networks where the eavesdroppers are not completely passive nodes but other users who are authorized but untrusted \cite{reference9}. They coexist in the same network being curious about the information exchanged between Alice and Bob, and their CSI can be obtained during their communications with the access point Alice. We note that in non-IRS systems, e.g, \cite{reference6, reference10}, only the direct channels ${{h}_{ab}}$, ${{h}_{ae_k}}$, and ${{h}_{be_k}}$ are considered. Since they are commonly modeled as Rayleigh channels, one of the eavesdropping channels $h_{E1_k}$ and $h_{E2_k}$ can be treated as independent to the other channels and neglected in the secret key capacity characterization due to the spatial decorrelation.

Considering that the IRS provides additional reflected channel terms to $h_{E1_k}$ and $h_{E2_k}$, which contain partially common information as the legitimate channels ${{h}_A}$ and ${{h}_B}$ (the reflecting coefficient matrix ${\bf{\Phi }}$ and the incident channel ${{\bf{h}}_{ar}}$ or ${{\bf{h}}_{br}}$) and dominate these channels, both $h_{E1_k}$ and $h_{E2_k}$ should be included in our analysis. To formulate the optimization problem and establish a relationship between the optimized variables and optimized performance, we first derive a new closed-form expression of the IRS secret key capacity lower bound as the worst case security scenario, and then design an efficient optimization algorithm to maximize this bound.

\begin{theorem}
For Eve $k$, the minimum achievable secret key capacity in IRS assisted wireless networks is
\begin{equation}\label{equation5}
{C_{{\rm{lb}},k}} = \frac{1}{{{T_c}}}{\log _2}\left(1 + \frac{{K_{ab}^2 - {K_{ab}}({K_{ae_k}} + {K_{be_k}})}}{{({K_{ae_k}} + {K_{be_k}} + {\sigma ^2})(2{K_{ab}}{\rm{ + }}{\sigma ^2})}}\right),
\end{equation}
where ${K_{pq}} = \mathbb{E}\{ ({h_{pq}} + {\bf{h}}_{rq}^H{\bf{\Phi }}{{\bf{h}}_{pr}}){({h_{pq}} + {\bf{h}}_{rq}^H{\bf{\Phi }}{{\bf{h}}_{pr}})^*}\}$ is the correlation function with $p \in \{a,b\}$, $q \in \{b,e_k\}$, and $\mathbb{E}\{\cdot\}$ is the expectation with respect to (w.r.t.) the $L$ samples in the normalized channel vectors.
\end{theorem}
\begin{proof}
In (\ref{equation4}), the lower bound is expressed as the maximum of ${C_{{\rm{lb1}},k}} \!=\! I({h_A};{h_B}) \!-\! I({h_A};{h_{E1_k}},{h_{E2_k}})$ and ${C_{{\rm{lb2}},k}} \!=\! I({h_A};{h_B}) \!-\! I({h_B};{h_{E1_k}},{h_{E2_k}})$, so we proceed to derive closed-form expressions for each term and compare them to obtain the final expression.

According to \cite{reference10}, the first term ${C_{{\rm{lb1}},k}}$ can be expressed as
\begin{equation}\label{equation6}
\begin{aligned}
{C_{{\rm{lb1}},k}}
&=I({h_A};{h_B}) - I({h_A};{h_{E1_k}},{h_{E2_k}})\\
&=H({h_A}{\rm{|}}{h_{E1_k}},\!{h_{E2_k}}) \!-\! H({h_A}{\rm{|}}{h_B})\\
&=\frac{1}{{{T_c}}}{\log _2}\frac{{\det ({{\bf{W}}_{AE1_kE2_k}}){K_{BB}}}}{{\det ({{\bf{W}}_{E1_kE2_k}})\det ({{\bf{W}}_{AB}})}},
\end{aligned}
\end{equation}
where $\det(\cdot)$ is the matrix determinant. In (\ref{equation6}), we adopt the approximation that $\mathbb{E} \{\log_2 (x)\} \approx \log_2 (\mathbb{E}\{x\})$ when $x$ is a variable with a small variance, which is valid in the high SNR region. The matrix term in the denominator of (\ref{equation6}) is given as
\begin{equation}\label{equation7}
\begin{aligned}
{{\bf{W}}_{AE1_kE2_k}} &\!=\! {\mathbb{E}\left\{{\left(\! {\begin{array}{*{20}{c}}
{{h_A}}\\
{{h_{E1_k}}}\\
{{h_{E2_k}}}
\end{array}}\!\right)({h_A^*})({h_{E1_k}^*})({h_{E2_k}^*})} \right\}}\\
& \!=\! \!{\left[ {\begin{array}{*{20}{l}}
{{K_{AA}}}&{{K_{AE1_k}}}&{{K_{AE2_k}}}\\
{{K_{AE1_k}}}&{{K_{E1_kE1_k}}}&{{K_{E1_kE2_k}}}\\
{{K_{AE2_k}}}&{{K_{E1_kE2_k}}}&{{K_{E2_kE2_k}}}
\end{array}} \right]},
\end{aligned}
\end{equation}
where ${K_{PQ}} = \mathbb{E}\{ {h_P}h_Q^*\}, P, Q \in \{ A,B,E1_k,E2_k\}$. Furthermore, we can obtain that ${K_{AA}} = {K_{ab}} + \sigma _a^2$, ${K_{BB}} = {K_{ab}} + \sigma _b^2$, ${K_{E1_kE1_k}} = {K_{ae_k}} + \sigma _{e_k}^2$, and ${K_{E2_kE2_k}} = {K_{be_k}} + \sigma _{e_k}^2$, where ${K_{pq}}$ is the correlation function. As such, we can derive a closed-form expression of the determinant for (\ref{equation7}) as
\begin{equation}\label{equation8}
\begin{aligned}
&\det({{{\bf{W}}_{AE1_kE2_k}}})\\
\!=&\: {K_{AA}}{K_{E1_kE1_k}}{K_{E2_kE2_k}}
\!+\!2{K_{AE1_k}}{K_{AE2_k}}{K_{E1_kE2_k}}\\
&\!-\! K_{AE2_k}^2{K_{E1_kE1_k}}
\!-\! K_{E1_kE2_k}^2{K_{AA}} \!-\! K_{AE1_k}^2{K_{E2_kE2_k}}\\
\!=&\: ({K_{ab}} \!+\! \sigma _a^2)({K_{a{e_k}}} \!+\! \sigma _{{e_k}}^2)({K_{b{e_k}}} \!+\! \sigma _{{e_k}}^2) \!+\! 2{K_{ab}}{K_{a{e_k}}}{K_{b{e_k}}}\\ &\!-\! {K_{ab}}{K_{b{e_k}}}({K_{a{e_k}}} \!+\! \sigma _{{e_k}}^2)
\!-\! {K_{a{e_k}}}{K_{b{e_k}}}({K_{ab}} \!+\! \sigma _a^2)\\ &\!-\! {K_{ab}}{K_{a{e_k}}}({K_{b{e_k}}} \!+\! \sigma _{{e_k}}^2)\\
\!=&\: ({K_{ae_k}} \!+\! {K_{be_k}})\sigma _a^2\sigma _{e_k}^2 \!+\! {K_{ab}}\sigma _{e_k}^4 \!+\! \sigma _a^2\sigma _{e_k}^4.
\end{aligned}
\end{equation}
Similarly, the determinants of the other matrices in (\ref{equation6}) can be calculated as
\begin{equation}\label{equation9}
\begin{aligned}
\det ({{\bf{W}}_{E{1_k}E{2_k}}})
&\!=\! {\rm{det}}\left( {\mathbb{E}\left\{ {\left( {\begin{array}{*{20}{c}}
{{h_{E{1_k}}}}\\
{{h_{E{1_k}}}}
\end{array}} \right)(h_{E{1_k}}^*)(h_{E{2_k}}^*)} \right\}} \right)\\
&\!=\! ({K_{a{e_k}}} + {K_{b{e_k}}})\sigma _{{e_k}}^2 + \sigma _{{e_k}}^4,
\end{aligned}
\end{equation}
\begin{equation}\label{equation10}
\begin{aligned}
\det({{{\bf{W}}_{AB}}})
&\!=\! {\rm{det}}\left( {\mathbb{E}\left\{ {\left( {\begin{array}{*{20}{c}}
{{h_A}}\\
{{h_B}}
\end{array}} \right)(h_A^*)(h_B^*)} \right\}} \right)\\
&\!=\! {K_{ab}}(\sigma _a^2 + \sigma _b^2) + \sigma _a^2\sigma _b^2.
\end{aligned}
\end{equation}
Substituting (\ref{equation8})--(\ref{equation10}) into (\ref{equation6}), we have ${C_{{\rm{lb1}},k}}$ as shown at the bottom of this page.

Likewise, we can derive a closed-form expression for the second term ${C_{{\rm{lb2}},k}}$, which is the same as (\ref{equation11}) except $\sigma _b^2$ in the numerator is replaced with $\sigma _a^2$. Therefore, ${C_{{\rm{lb}},k}}={C_{{\rm{lb1}},k}}$ if $\sigma _b^2 \le \sigma _a^2$, otherwise ${C_{{\rm{lb}},k}}={C_{{\rm{lb2}},k}}$. When we consider all noises are equal, i.e., $\sigma _a^2 = \sigma _b^2 = \sigma _e^2 = {\sigma ^2}$, both of them can be simplified as (\ref{equation5}), which completes the proof.
\end{proof}

\section{Proposed SCA-SDR Optimization}
In this section, we develop a low-complexity optimization framework to determine the $N$-element reflecting coefficient matrix ${\bf{\Phi}}$ for Rose that maximizes the minimum secret key capacity amongst all the $K$ non-colluding eavesdroppers. We note that the optimization is commonly performed at the access point Alice and the derived optimal solution should be sent to Rose within the coherence time. If Rose is equipped with high-cost hardware and provided with sufficient computing resources, this processing can also be undertaken by the IRS, which leads to faster adjustments.

Firstly, the optimization problem is formulated as
\setcounter{equation}{11}
\begin{subequations}\label{equation12}
\begin{align}
\mathop {\max_{\bm{\varphi}}}{\kern 1pt}{\kern 1pt}{\min_{k \in \mathcal{K}}}{\kern 1pt} {\kern 1pt} {\kern 1pt} {\kern 1pt} &{C_{{\rm{lb}},k}} \\
{\rm{s.}}{\rm{t.}}{\kern 1pt} {\kern 1pt}{\kern 1pt}{\kern 1pt} &{\left| {{\varphi _n}} \right|^2} \le 1, n = 1, \ldots ,N,
\end{align}
\end{subequations}
where $\mathcal{K}=\{1, \ldots ,K\}$, $\left| \cdot \right|$ is the modulus operator, and ${\varphi _n} = {\alpha _n}{e^{j{\theta _n}}}$ is the $n$th diagonal element of ${\bf{\Phi}}$. The constraints in (\ref{equation12}b) applies for ${\alpha _n} \in [0,1]$.

In our derived closed-form expression for the secret key capacity in (\ref{equation5}), we note that $T_c$ is a constant for a given channel and $\log _2(\cdot)$ is a monotonically increasing function. As such, the optimization problem (\ref{equation12}) can be rewritten as
\begin{subequations}\label{equation13}
\begin{align}
\mathop {\max_{\bm{\varphi}}}{\kern 1pt}{\kern 1pt}{\min_{k \in \mathcal{K}}}{\kern 1pt} {\kern 1pt} {\kern 1pt} {\kern 1pt}  &\left(1 + \frac{{K_{ab}^2 - {K_{ab}}({K_{ae_k}} + {K_{be_k}})}}{{({K_{ae_k}} + {K_{be_k}} + {\sigma ^2})(2{K_{ab}}{\rm{ + }}{\sigma ^2})}}\right)\\
{\rm{s.}}{\rm{t.}}{\kern 1pt} {\kern 1pt}{\kern 1pt}{\kern 1pt} &{\left| {\varphi_n} \right|^2} \le 1, n = 1, \ldots ,N.
\end{align}
\end{subequations}
To establish a direct optimization relationship with the IRS coefficient matrix ${\bf{\Phi}}$, we need to calculate $K_{ab}$, $K_{ae_k}$, and $K_{be_k}$. To do so, we define ${{\bf{v}}^H} = [{v_1},{v_2}, \ldots ,{v_N}]$, where ${v_n} = \varphi _n^*$, and ${{\bf{\Sigma }}_{pr}} = {\rm{diag(}}{\bf{h}}_{pr}^H{\rm{)}}$, then by rearranging the positions of the variables as ${\bf{h}}_{pr}^H{\bf{\Phi }}{{\bf{h}}_{rq}}{\bf{h}}_{rq}^H{{\bf{\Phi }}^H}{\bf{h}}_{pr} = {{\bf{v}}^H}{{\bf{\Sigma }}_{pr}}{{\bf{h}}_{rq}}{\bf{h}}_{rq}^H{\bf{\Sigma }}_{pr}^H{\bf{v}}$, we can derive the expressions of $K_{ab}$, $K_{ae_k}$, and $K_{be_k}$ as
\begin{equation}\label{equation14}
\begin{array}{l}
{K_{ab}} = \sigma _{{h_{ab}}}^2 + {{\bf{v}}^H}{{\bf{R}}_{arb}}{\bf{v}},\\
{K_{a{e_k}}} = \sigma _{{h_{a{e_k}}}}^2 + {{\bf{v}}^H}{{\bf{R}}_{ar{e_k}}}{\bf{v}},\\
{K_{b{e_k}}} = \sigma _{{h_{b{e_k}}}}^2 + {{\bf{v}}^H}{{\bf{R}}_{br{e_k}}}{\bf{v}},
\end{array}
\end{equation}
where $\sigma _{{h_{pq}}}^2 = \mathbb{E}\{ {h_{pq}}h_{pq}^*\}$ and ${{\bf{R}}_{prq}} = \mathbb{E}\{ {{\bf{\Sigma }}_{pr}}{{\bf{h}}_{rq}}{\bf{h}}_{rq}^H{\bf{\Sigma }}_{pr}^H\}$, which is a positive semi-definite covariance matrix of combined channel vectors. Based on (\ref{equation14}), we can further expand the numerator and denominator of (\ref{equation13}a) as
\begin{equation}\label{equation15}
\begin{aligned}
h_k({\bf{v}}) =&\: {a_1} + {b_1}{{\bf{v}}^H}{{\bf{R}}_{arb}}{\bf{v}} + {c_1}{{\bf{v}}^H}{{\bf{R}}_{\Sigma}}{\bf{v}}\\
&+ {{\bf{v}}^H}{{\bf{R}}_{arb}}{\bf{v}}{{\bf{v}}^H}{{\bf{R}}_{arb}}{\bf{v}} +{{\bf{v}}^H}{{\bf{R}}_{arb}}{\bf{v}}{{\bf{v}}^H}{{\bf{R}}_{\Sigma}}{\bf{v}},\\
\end{aligned}
\end{equation}
\begin{equation}\label{equation16}
g_k({\bf{v}}) \!=\! {a_2} \!+\! {b_2}{{\bf{v}}^H}{{\bf{R}}_{arb}}{\bf{v}} \!+\! {c_2}{{\bf{v}}^H}{{\bf{R}}_{\Sigma}}{\bf{v}} \!+\! 2{{\bf{v}}^H}{{\bf{R}}_{arb}}{\bf{v}}{{\bf{v}}^H}{{\bf{R}}_{\Sigma}}{\bf{v}},
\end{equation}
where we set ${{\bf{R}}_\Sigma } \!=\! {{\bf{R}}_{ar{e_k}}} \!+\! {{\bf{R}}_{br{e_k}}}$, and $\sigma _\Sigma ^2 \!=\! \sigma _{{h_{a{e_k}}}}^2 \!+\! \sigma _{{h_{b{e_k}}}}^2$, and the scalar terms in (\ref{equation15}) and (\ref{equation16}) are ${a_1} \!=\! \sigma _{{h_{ab}}}^4 \!+\! \sigma _{{h_{ab}}}^2\sigma _\Sigma ^2 \!+\! 2{\sigma ^2}\sigma _{{h_{ab}}}^2 \!+\! {\sigma ^2}\sigma _\Sigma ^2 \!+\! {\sigma ^4}$, ${b_1} \!=\! 2\sigma _{{h_{ab}}}^2 \!+\! \sigma _\Sigma ^2 \!+\! 2{\sigma ^2}$, ${c_1} \!=\! \sigma _{{h_{ab}}}^2 \!+\! {\sigma ^2}$, ${a_2} \!=\! 2\sigma _{{h_{ab}}}^2\sigma _\Sigma ^2 \!+\! {\sigma ^2}\sigma _\Sigma ^2 \!+\! 2{\sigma ^2}\sigma _{{h_{ab}}}^2 \!+\! {\sigma ^4}$, ${b_2} \!=\! 2\sigma _\Sigma ^2 \!+\! 2{\sigma ^2}$, and ${c_2} \!=\! 2\sigma _{{h_{ab}}}^2 \!+\! {\sigma ^2}$.

We note that the objective function in (\ref{equation13}a) is non-convex due to the max-min operations. In the following, we introduce an auxiliary variable $C$ to transform the optimization based on (\ref{equation15}) and (\ref{equation16}) as
\begin{subequations}\label{equation17}
\begin{align}
\mathop {\max_{{\bf{v}},C}}{\kern 1pt}{\kern 1pt}{\kern 1pt} {\kern 1pt}  &C\\
{\rm{s.}}{\rm{t.}}{\kern 1pt} {\kern 1pt}{\kern 1pt}{\kern 1pt} &C \in {\mathbb{R}},\\
&{\left| {v_n} \right|^2} \le 1, n = 1, \ldots ,N,\\
&{h_k}({\bf{v}}) \ge {C}{g_k}({\bf{v}}), k = 1, \ldots ,K.
\end{align}
\end{subequations}

Due to the non-convexity of the matrix product terms in (\ref{equation15}) and (\ref{equation16}), the constraints in (\ref{equation17}d) are still non-convex w.r.t. the optimization variable ${\bf{v}}$, where denoting ${\bf{V}} = {\bf{v}}{{\bf{v}}^H}$ confirms that ${\bf{V}} \succeq 0$ and ${\rm{rank}}({\bf{V}}) = 1$. Since the rank-1 constraint is non-convex, we proceed to apply the SDR technique to relax this constraint. By substituting ${{\bf{v}}^H}{\bf{Av}}{{\bf{v}}^H}{\bf{Bv}} = {\rm{Tr}}({\bf{AVBV}})$ and ${{\bf{v}}^H}{\bf{Cv}} = {\rm{Tr}}({\bf{CV}})$, where $\bf{A}$, $\bf{B}$, and $\bf{C}$ are any positive semi-definite matrices and ${\rm{Tr}}(\cdot)$ is the trace of a matrix, we can now consider that $h_k({\bf{V}})$ and $g_k({\bf{V}})$ are both convex w.r.t. ${\bf{V}}$. Finally, the constraints in (\ref{equation17}d) are all in forms of the Difference of Convex (DC) functions for a given $C$, which can be globally optimized by standard techniques such as branch-and-bound and cutting planes algorithms. To further reduce the complexity of the optimization, we propose an efficient suboptimal solution based on the SCA technique where we apply the first order Taylor series expansion at ${\bf{V}}^{(m)}$ to obtain a linear approximation, which results in
\begin{equation}\label{equation18}
{h_k}({\bf{V}}) \!=\! {h_k}({\bf{V}}^{(m)}) \!+\! {\rm{Tr}}({\rm{Re}}\{ \nabla {h_k}{({\bf{V}}^{(m)})^H}({\bf{V}} \!-\! {\bf{V}}^{(m)})\}),
\end{equation}
where the gradient of ${h_k}({\bf{V}})$ at ${\bf{V}}^{(m)}$ is given as
\begin{equation}\label{equation19}
\begin{aligned}
\nabla {h_k}({{\bf{V}}^{(m)}}) =&\: {b_1}{\bf{R}}_{arb}^T \!+\! {c_1}{\bf{R}}_\Sigma ^T \!+\! 2{({{\bf{R}}_{arb}}{{\bf{V}}^{(m)}}{{\bf{R}}_{arb}})^T}\\
&+ {({{\bf{R}}_{arb}}{{\bf{V}}^{(m)}}{{\bf{R}}_\Sigma } + {{\bf{R}}_\Sigma }{{\bf{V}}^{(m)}}{{\bf{R}}_{arb}})^T}.
\end{aligned}
\end{equation}
Moreover, we introduce a base vector ${{\bf{e}}_n}$, whose $n$th element is 1 and others 0, to express the constraints in (\ref{equation17}c) with ${\bf{V}}$, which can be shown as ${\left| {{v_n}} \right|^2} = {\left| {{{\bf{v}}^H}{{\bf{e}}_n}} \right|^2} = {{\bf{v}}^H}{{\bf{e}}_n}{\bf{e}}_n^H{\bf{v}} = {\rm{Tr}}({{\bf{E}}_n}{\bf{V}}) \le 1$, where ${{\bf{E}}_n}={{\bf{e}}_n}{\bf{e}}_n^H$. As a result, the optimization problem can be re-expressed as
\begin{subequations}\label{equation20}
\begin{align}
\mathop {\max_{{\bf{V}},C}}{\kern 1pt}{\kern 1pt}{\kern 1pt} {\kern 1pt}  &C\\
{\rm{s.}}{\rm{t.}}{\kern 1pt} {\kern 1pt}{\kern 1pt}{\kern 1pt} &C \in {\mathbb{R}}, {\bf{V}} \succeq 0,\\
&{\rm{Tr}}({{\bf{E}}_n}{\bf{V}}) \le 1, n = 1, \ldots ,N,\\
&{C}{g_k}({\bf{V}})- {h_k}({\bf{V}}) \le 0, k = 1, \ldots ,K.
\end{align}
\end{subequations}

Therefore, we have successfully convexified the IRS optimization problem in (\ref{equation20}) which can be solved by alternatively optimizing $\bf{V}$ and $C$. Algorithm 1 details the proposed SDR-SCA based algorithm, where $\epsilon$ denotes a small convergence threshold, $M$ is the maximum number of iterations, and $C_{\max}$ is a sufficiently large number. To extract a rank one solution from the optimal matrix $\bf V$, the well-known approach of Gaussian randomization \cite{reference11} is employed. Note that the proposed SDR-SCA based algorithm can solve the problem with a worst case complexity of $\mathcal{O}(\max\{N,K\}^4{N^{{1 \mathord{\left/ {\vphantom {1 2}} \right. \kern-\nulldelimiterspace} 2}}}{\bar M}\log_2({C_{\max} \mathord{\left/ {\vphantom {C_{\max} \epsilon }} \right. \kern-\nulldelimiterspace} \epsilon }))$, where ${\bar M}$ is the average number of iterations. This is a polynomial complexity algorithm, which is more efficient than standard branch-and-bound algorithm with exponential complexity of $\mathcal{O}({({1 \mathord{\left/ {\vphantom {1 {{\delta _{\rm{A}}}}}} \right. \kern-\nulldelimiterspace} {{\delta _{\rm{A}}}}}{\delta _{\rm{P}}})^{KN}})$, where ${{\delta _{\rm{A}}}}$ and ${{\delta _{\rm{P}}}}$ are the discrete quantization intervals of amplitude and phase.

\begin{algorithm}[tb]
\caption{Proposed SDR-SCA based Iterative Optimization}
\label{algorithm1}
\begin{algorithmic}[1]
\renewcommand{\algorithmicrequire}{\textbf{Input:}}
\Require
${\bf{R}}_{arb}$, ${\bf{R}}_\Sigma$, $\sigma _{{h_{ab}}}^2$, $\sigma _\Sigma ^2$, $\sigma^2$, $K$, $\epsilon$, $M$, and $C_{\max}$.
\renewcommand{\algorithmicensure}{\textbf{Output:}}
\Ensure
$\bf{v}$.
\State Initialize $C_{\min}=0$, and set $C^{(t)}={{({C_{\max }} + {C_{\min }})} \mathord{\left/ {\vphantom {{({C_{\max }} + {C_{\min }})} 2}} \right. \kern-\nulldelimiterspace} 2}$.
\Repeat \;{(Bisection search for $C$)}
\Repeat \;{(SDR-SCA algorithm for $\bf{V}$)}
\State For given $C^{(t)}$, when $m=1$, initialize a positive semi-definite reflecting coefficient matrix ${\bf{V}}^{(1)}$ randomly; when $m>1$, given ${\bf{V}}^{(m-1)}$, find the optimal optimization variable ${\bf{V}}^{(m)}$ according to the problem (\ref{equation19}).
\State Update $m=m+1$.
\Until the optimization variable ${\bf{V}}$ reaches convergence or $m=M$.
\If {the aforementioned problem is solvable}
\State record ${\bf{V}}_{\rm{opt}}={\bf{V}}^{(m)}$ of the current iteration, then update $C_{\min}=C^{(t)}$;
\Else \;{$C^{(t)}$ is unreachable, then update $C_{\max}=C^{(t)}$}.
\EndIf
\Until the difference $({C_{\max }}-{C_{\min }})$ is below $\epsilon$.
\State Recover $\bf{v}$ from ${\bf{V}}_{\rm{opt}}$ by Gaussian randomization \cite{reference11}.
\end{algorithmic}
\end{algorithm}

\section{Simulation Results}
This section presents simulation results to highlight the performance advantage of IRS for secret key generation in wireless networks. We consider the network scenario in Fig.~\ref{figure2}, where Alice, Bob, and the central point of Rose are located at (5, 0, 20), (3, 100, 0), and (0, 100, 2), respectively. For eavesdroppers, we consider two cases: (1) $K$ Eves are randomly distributed in a circle of radius $R$ centered on Bob, (2) $K$ Eves are randomly distributed in a circle of radius $R$ centered on Alice. The simulation parameters are shown in Table~\ref{table1}, if not specifically mentioned.
\begin{figure}[t]
  \centering
  \includegraphics[width=0.4\textwidth]{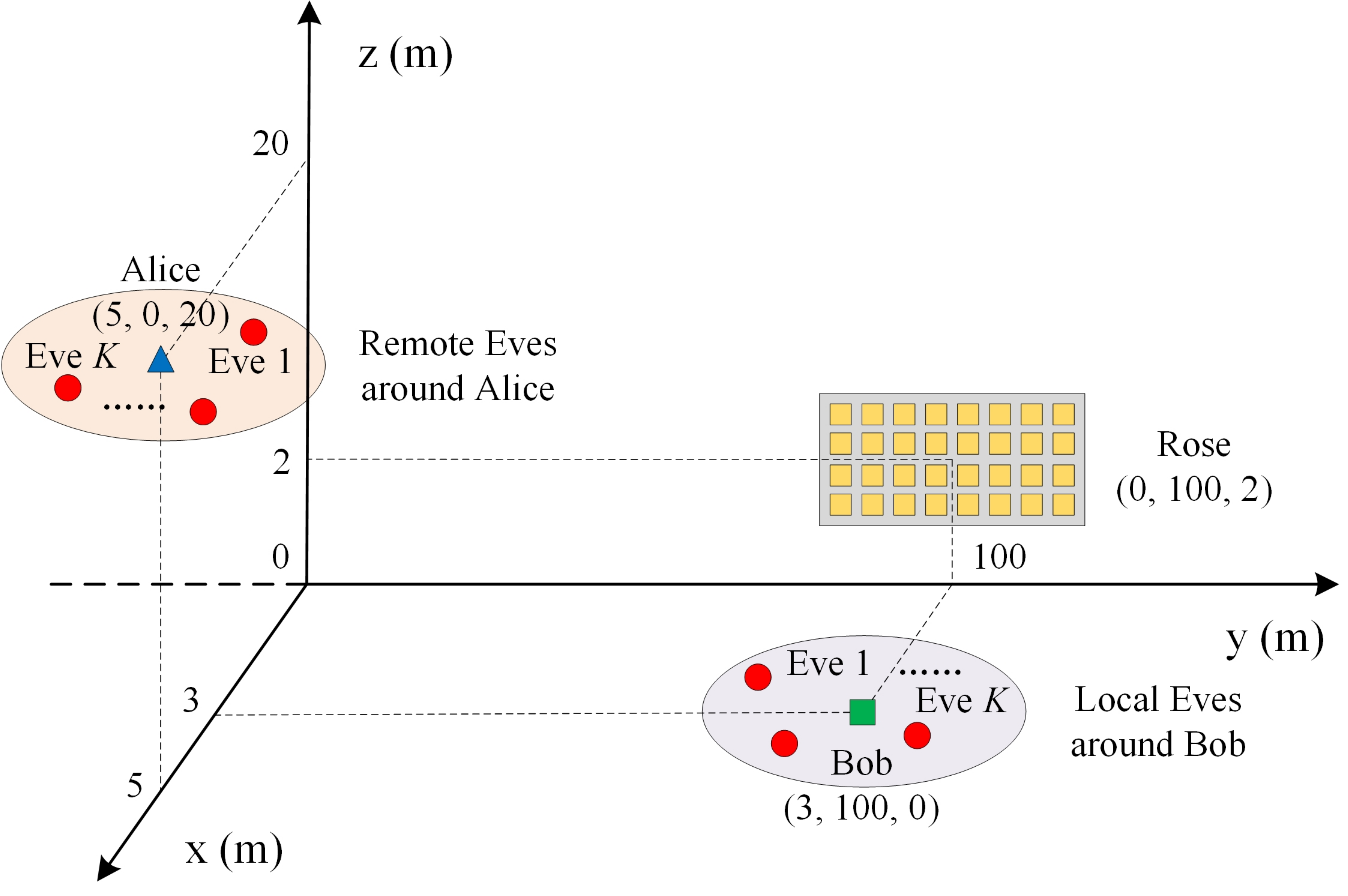}\\
  \caption{Simulation setup: (1) Eves around Bob and (2) Eves around Alice.}
  \label{figure2}
\end{figure}

\begin{table}[t]
\renewcommand{\arraystretch}{1.1}
\caption{Simulation Parameters}
\label{table1}
\centering
\begin{tabular}{|l|l|}
\hline
\bf{Parameter}&\bf{Value}\\
\hline
Carrier frequency&1 GHz.\\
\hline
Path loss at 1 m&${\zeta _0}$ = --30 dB.\\
\hline
\multirow{2}*{IRS configuration}&Uniform rectangular array (URA) with 5 rows\\& and $N/5$ columns, $\lambda$ spacing.\\
\hline
\multirow{2}*{Path loss exponent}&a) ${c_{ab}}$ \!=\! ${c_{ae_k}}$=\! 5, ${c_{ar}}$ \!=\! 3.5, ${c_{be_k}}$=\! ${c_{rb}}$ \!=\! ${c_{re_k}}$=\! 2;\\
&b) ${c_{ab}}$ \!=\! ${c_{be_k}}$=\! 5, ${c_{ar}}$ \!=\! ${c_{re_k}}$=\! 3.5, ${c_{be_k}}$= ${c_{rb}}$ \!=\! 2.\\
\hline
\multirow{2}*{Rician factor}&a) $\kappa _{ab} \!=\! \kappa _{ae_k}$=\! 3, ${\kappa_{ar}}$ \!=\! ${\kappa_{rb}}$ \!=\! ${\kappa_{re_k}}$=\! 2, ${\kappa_{be_k}}$=\! 5;\\
&b) $\kappa _{ab} \!=\! \kappa _{be_k}$=\! 3, ${\kappa_{ar}}$ \!=\! ${\kappa_{rb}}$ \!=\! ${\kappa_{re_k}}$=\! 2, ${\kappa_{ae_k}}$=\! 5.\\
\hline
Other parameters&$\sigma^2$ = --105 dB, ${T_c}$ = 1 s, $\epsilon$ = 0.01, $L$ = $M$ = 50.\\
\hline
\end{tabular}
\end{table}

The direct channel between Alice and Bob ${\tilde h_{ab}}$ is generated by ${\tilde h_{ab}} = \sqrt {{\zeta _0}d_{ab}^{ - {c_{ab}}}} {g_{ab}}$, where $d_{ab}$ and ${c_{ab}}$ are the distance from Alice to Bob, and the path loss exponent, respectively. The small-scale fading component ${g_{ab}}$ is assumed to be Rician fading defined as ${g_{ab}} = \sqrt {{{{\kappa _{ab}}} \mathord{\left/  {\vphantom {{{\kappa _{ab}}} {(1 + {\kappa _{ab}})}}} \right. \kern-\nulldelimiterspace} {(1 + {\kappa _{ab}})}}} g_{ab}^{{\rm{LoS}}} + \sqrt {{1 \mathord{\left/ {\vphantom {1 {(1 + {\kappa _{ab}})}}} \right. \kern-\nulldelimiterspace} {(1 + {\kappa _{ab}})}}} g_{ab}^{{\rm{NLoS}}}$, where ${\kappa _{ab}}$ is the Rician factor, while $g_{ab}^{{\rm{LoS}}}$ and $g_{ab}^{{\rm{NLoS}}}$ are the deterministic LoS and Rayleigh non-LoS (NLoS) components. The same channel model is employed for all other channels ${{\tilde h}_{a{e_k}}}, {{\tilde h}_{b{e_k}}}, {{{\bf{\tilde h}}}_{ar}}, {{{\bf{\tilde h}}}_{br}}$, and ${{{\bf{\tilde h}}}_{r{e_k}}}$. Since Eves are assumed to be located around either Alice or Bob, the Pearson correlation coefficient $\rho  = [{J_0}({{2\pi d} \mathord{\left/ {\vphantom {{2\pi d} \lambda }} \right.
\kern-\nulldelimiterspace} \lambda })]^2$ is considered between the NLoS components at locations separated by distance
$d$, where ${J_0}(\cdot)$ is the Bessel function of the first kind, and $\lambda$ is the wavelength. We note that for secret key generation, the effect of large-scale fading is removed by using sample normalization to ensure a large uncorrelation of channels among spatially distributed nodes and thus the randomness of generated keys. In simulations, we consider two benchmarks of without IRS and IRS with random shifting for the reflecting coefficients.

Fig. \ref{figure3} shows the secret key capacity versus the number of IRS elements, $N$, for case (1), where Eves are located around Bob. We see that increasing the size of the IRS results in a significant improvement in the secret key capacity of our SDR-SCA scheme. Furthermore, we observe that SDR-SCA outperforms the two benchmark schemes for the entire range of $N$. The figure also shows that the IRS secret key capacity increases with decreasing $\kappa = {\kappa _{rb}} \!=\! {\kappa _{r{e_k}}}$ which corresponds to a weaker LoS path. This is because a lower $\kappa$ provides more randomness in the wireless channels leading to lower correlations between the legitimate and eavesdropper channels. If the channel statistics are unavailable, the IRS with random shifting only results in a small performance improvement compared to the case without IRS.
\begin{figure}[t]
  \centering
  \includegraphics[width=0.38\textwidth]{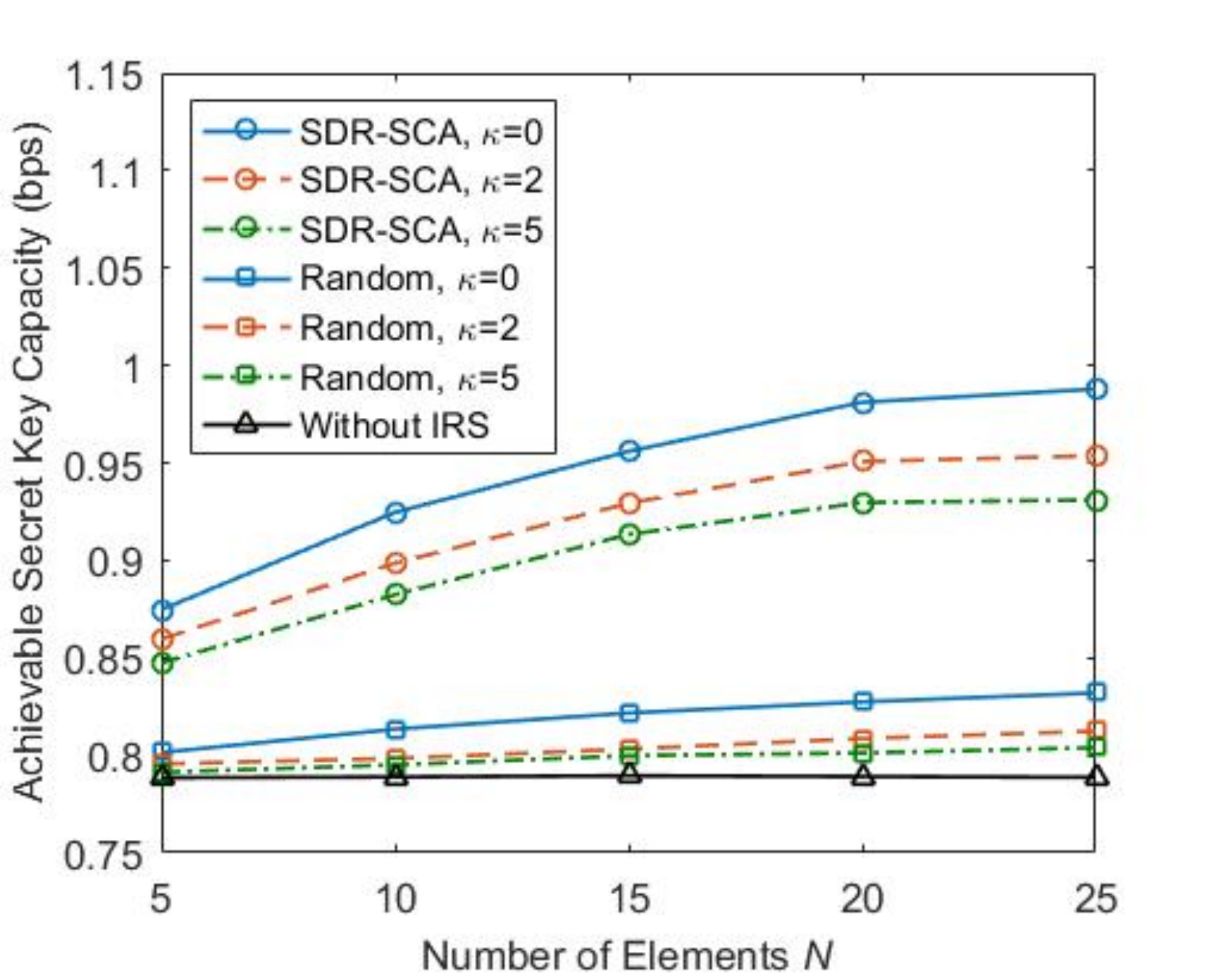}\\
  \caption{Achievable secret key capacity versus the number of IRS reflecting elements $N$ in case (1), under different Rician factors with $K$ = 2, $R$ = $\lambda$.}
  \label{figure3}
\end{figure}

In Fig. \ref{figure4}, we see that the achievable secret key capacity decreases with increasing number of Eves. We also find that the proposed SDR-SCA scheme significantly outperforms the two benchmark schemes by at least 0.1 bps. Compared with Fig. 3, we observe that the increasing number of Eves only has a small negative impact on the performance. As expected, the secret key capacity improves when there is a greater average distance between the eavesdroppers with $R$ = $\lambda$.
\begin{figure}[t]
  \centering
  \includegraphics[width=0.38\textwidth]{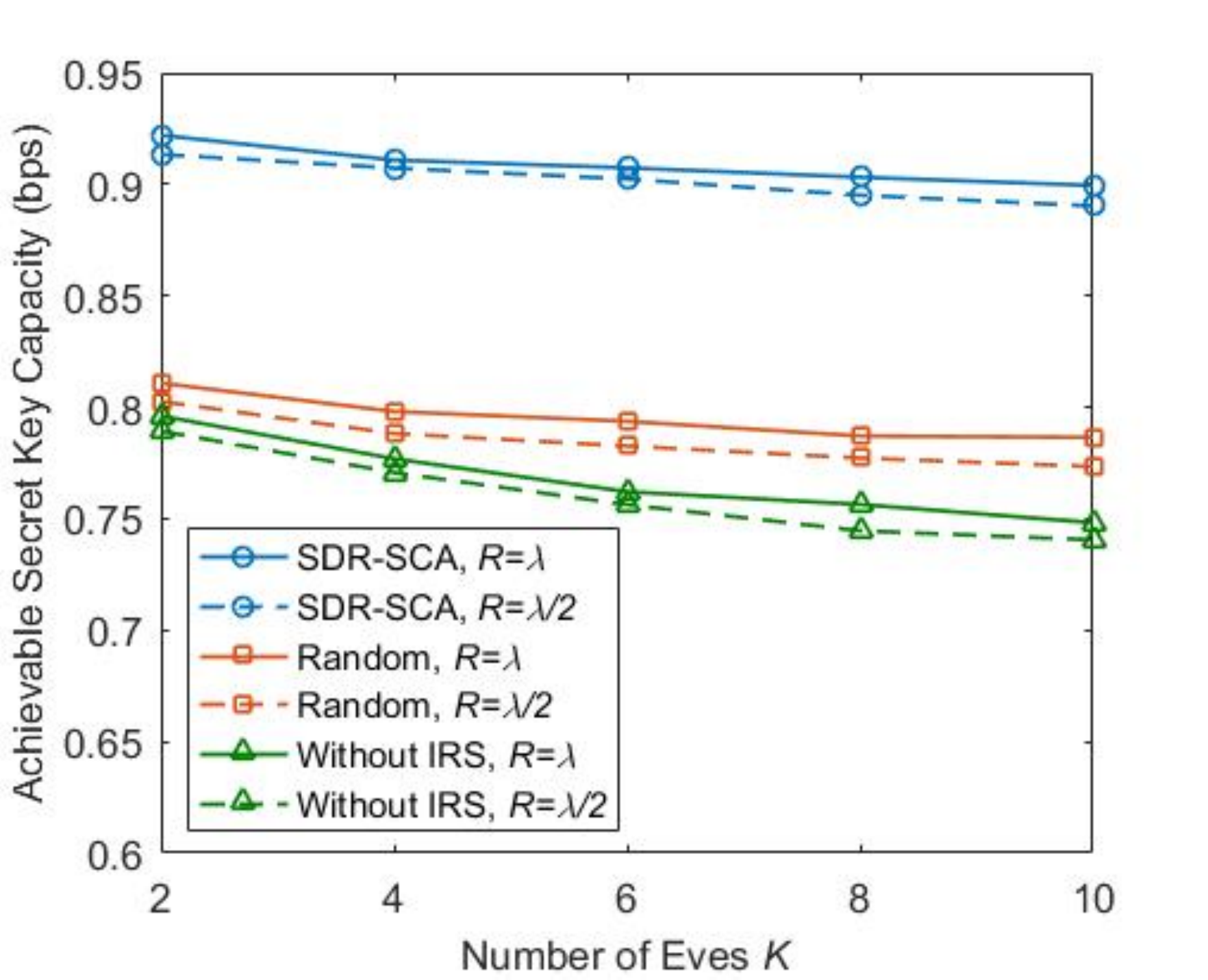}\\
  \caption{Achievable secret key capacity versus the number of Eves $K$ in case (1), under different distribution radiuses of Eves with $N$ = 20.}
  \label{figure4}
\end{figure}

In Fig. \ref{figure5}, we compare the secret key capacity of cases (1) and (2). The figure shows that the location of the eavesdroppers has a significant impact on the IRS secret key capacity whilst there is negligible impact on the non-IRS scheme. Specifically, we highlight that IRS can achieve a higher secret key capacity in case (1) when Eves are located around Bob. This is because in case (1), the IRS is also located closer to Eves and thus the IRS reflecting components will have a higher contribution to the eavesdropping channels compared to case (2) where the IRS is further away from the Eves. When the IRS is closer to Eves, the correlations ${K_{ae_k}}$ and ${K_{be_k}}$ in our derived capacity expression in (5) can be significantly reduced by our proposed IRS optimization, which leads to a great improvement in the secret key capacity.

\begin{figure}[t]
  \centering
  \includegraphics[width=0.38\textwidth]{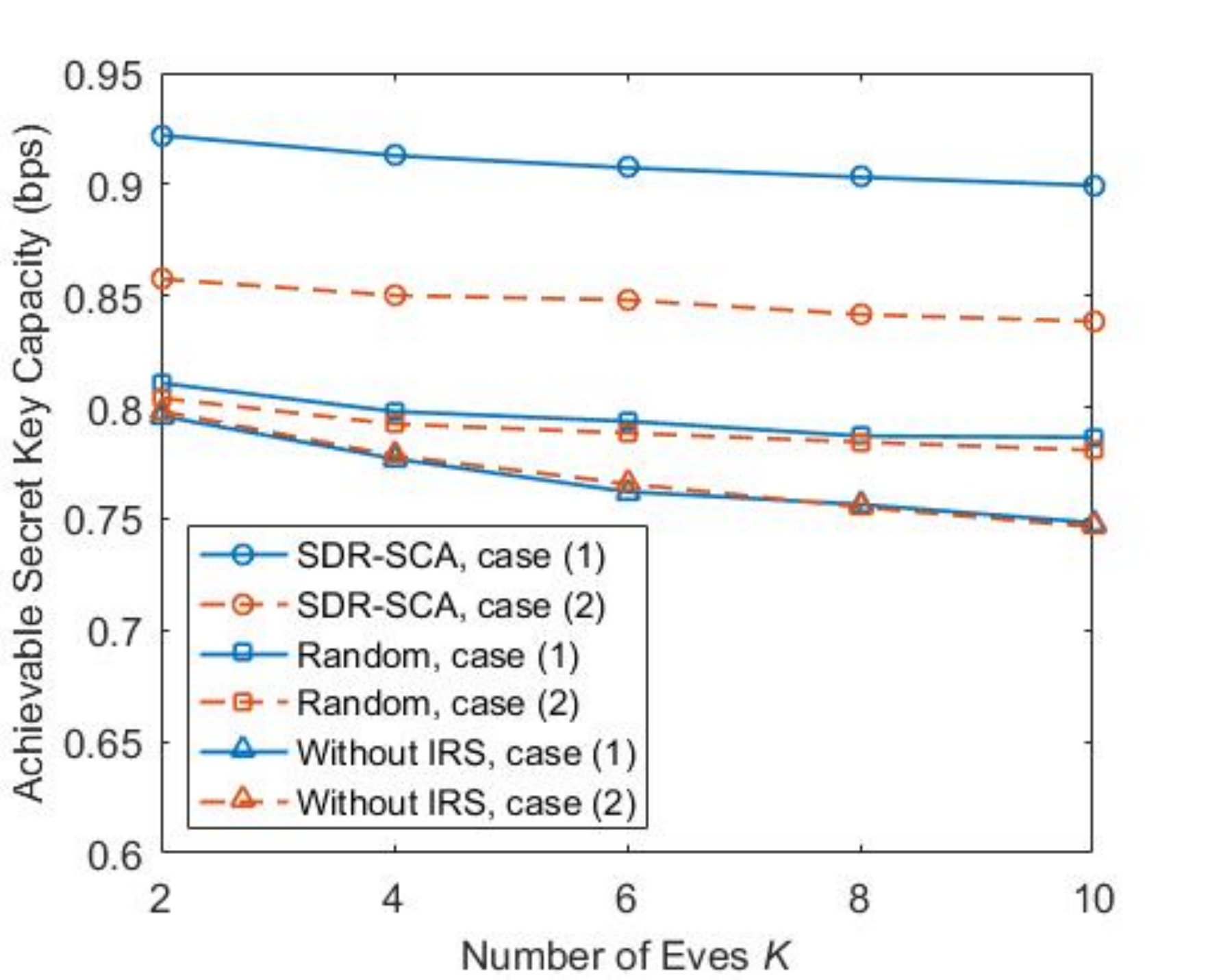}\\
  \caption{Achievable secret key capacity versus the number of Eves $K$ in both cases (1) and (2) with $N$ = 20, $R$ = $\lambda$.}
  \label{figure5}
\end{figure}

\section{Conclusions}
We derived a new lower bound on the secret key capacity of IRS assisted wireless networks with multiple non-colluding eavesdroppers. Based on this bound, an efficient SDR-SCA optimization algorithm was proposed to design the IRS reflecting coefficient matrix that maximizes the minimum achievable secret key capacity for the worst-case eavesdropper. Simulations showed that the achievable secret key capacity can be efficiently improved with our proposed IRS optimization algorithm for different eavesdropper locations and LoS channels conditions.


\begin{thebibliography}{11}
\bibitem{reference1}
C. Pan \emph{et al}., ``Multicell MIMO communications relying on intelligent reflecting surfaces,'' \emph{IEEE Trans. Wireless Commun.}, vol.~19, no.~8, pp.~5218--5233, Aug.~2020.
\bibitem{reference1.1}
C. Pan \emph{et al}., ``Intelligent reflecting surface aided MIMO broadcasting for simultaneous wireless information and power transfer,'' \emph{IEEE J. Sel. Areas Commun.}, vol.~38, no.~8, pp.~1719--1734, Aug.~2020.
\bibitem{reference2}
Q. Wu and R. Zhang, ``Towards smart and reconfigurable environment:
Intelligent reflecting surface aided wireless network,'' \emph{IEEE Commun. Mag.}, vol.~58, no.~1, pp.~106--112, Jan.~2020.
\bibitem{reference3}
J. Chen, Y. Liang, Y. Pei, and H. Guo, ''Intelligent reflecting surface: A programmable wireless environment for physical layer security," \emph{IEEE Access}, vol.~7, pp.~82599--82612, Jun.~2019.
\bibitem{reference4}
X. Guan, Q. Wu, and R. Zhang, ``Intelligent reflecting surface assisted secrecy communication: Is artificial noise helpful or not?,'' \emph{IEEE Wireless Commun. Lett.}, vol.~9, no.~6, pp.~778--782, Jun.~2020.
\bibitem{reference4.1}
S. Hong, C. Pan, H. Ren, K. Wang, and A. Nallanathan, ``Artificial-noise-aided secure MIMO wireless communications via intelligent reflecting surface,'' \emph{IEEE Trans. Commun.}, DOI: 10.1109/TCOMM.2020.3024621, Sep.~2020.
\bibitem{reference5}
L. Jiao, N. Wang, P. Wang, A. Alipour-Fanid, J. Tang, and K. Zeng, ``Physical layer key generation in 5G wireless networks,'' \emph{IEEE Wireless Commun.}, vol.~26, no.~5, pp.~48--54, Oct.~2019.
\bibitem{reference6}
C. D. T. Thai, J. Lee, and T. Q. S. Quek, ``Physical-layer secret key generation with colluding untrusted relays,'' \emph{IEEE Trans. Wireless Commun.}, vol.~15, no.~2, pp.~1517--1530, Feb.~2016.
\bibitem{reference7}
Z. Ji \emph{et al}., ``Vulnerabilities of physical layer secret key generation against environment reconstruction based attacks,'' \emph{IEEE Wireless Commun. Lett.}, vol.~9, no.~5, pp.~693--697, May~2020.
\bibitem{reference8}
S. T. Ali, V. Sivaraman, and D. Ostry, ``Eliminating reconciliation cost in secret key generation for body-worn health monitoring devices,'' \emph{IEEE Trans. Mobile Comput.}, vol.~13, no.~12, pp.~2763--2776, Dec.~2014.
\bibitem{reference9}
G. Chen, Y. Gong, P. Xiao, and J. A. Chambers, ``Physical layer network security in the full-duplex relay system,'' \emph{IEEE Trans. Inf. Forensics Security}, vol.~10, no.~3, pp.~574--583, Mar.~2015.
\bibitem{reference10}
J. W. Wallace and R. K. Sharma, ``Automatic secret keys from reciprocal MIMO wireless channels: Measurement and analysis,'' \emph{IEEE Trans. Inf. Forensics Security}, vol.~5, no.~3, pp.~381--392, Sep.~2010.
\bibitem{reference11}
Z. Luo, W. Ma, A. M. So, Y. Ye, and S. Zhang, ``Semidefinite relaxation of quadratic optimization problems,'' \emph{IEEE Signal Process. Mag.}, vol.~27, no.~3, pp.~20--34, May~2010.
\end{thebibliography}
\end{document}